\documentclass[12pt,final,
thmsa,
epsf,a4paper]
{article}
\usepackage{booktabs}

\newtheorem{theorem}{Theorem}

\newtheorem{definition}{Definition}
\newtheorem{example}{Example}
\newtheorem{proposition}{Proposition}
\newtheorem{remark}{Remark}

\newenvironment{proof}[1][Proof]{\emph{#1.} }{\  \hfill $\square $ \vspace{5 pt}}

\usepackage{natbib}
\usepackage{amssymb}
\usepackage[dvips]{graphics}
\usepackage{amsmath}

\usepackage{mathpazo}
\usepackage{enumerate}

\usepackage{amsfonts}

\usepackage{amssymb}

\usepackage{enumerate}

\usepackage{mathrsfs}

\usepackage[usenames]{color}
\usepackage{cancel}

\usepackage{pgf,tikz}

\usetikzlibrary{decorations.pathreplacing,angles,quotes}
\usetikzlibrary{arrows.meta}
\usetikzlibrary{arrows, decorations.markings}
\tikzset{myptr/.style={decoration={markings,mark=at position 1 with %
       {\arrow[scale=2,>=stealth]{>}}},postaction={decorate}}}

\usepackage{hyperref}
\hypersetup{
    colorlinks,
    citecolor=blue,
    filecolor=black,
    linkcolor=blue,
    urlcolor=blue
}

\usepackage{pgf,tikz}
\usetikzlibrary{arrows}

\usepackage[utf8]{inputenc}
\usepackage{amssymb, amsmath,geometry}
\usepackage{fancyhdr}

\usepackage{soul}

\usepackage{emptypage}

\newcommand*\samethanks[1][\value{footnote}]{\footnotemark[#1]}

\usepackage[T1]{fontenc}
\DeclareFontFamily{T1}{calligra}{}
\DeclareFontShape{T1}{calligra}{m}{n}{<->s*[1.44]callig15}{}
\DeclareMathAlphabet\mathcalligra   {T1}{calligra} {m} {n}

\begin{document}

\title{Core and stability notions in many-to-one matching markets with indifferences\thanks{
 We acknowledge financial support
from UNSL through grants 032016 and 030320, from Consejo Nacional
de Investigaciones Cient\'{\i}ficas y T\'{e}cnicas (CONICET) through grant
PIP 112-200801-00655, and from Agencia Nacional de Promoción Cient\'ifica y Tecnológica through grant PICT 2017-2355.}}


\author{Agustín G. Bonifacio\thanks{Instituto de Matem\'{a}tica Aplicada San Luis, Universidad Nacional de San
Luis and CONICET, San Luis, Argentina, and RedNIE. Emails: \texttt{abonifacio@unsl.edu.ar} (A. G. Bonifacio), 
\texttt{nmjuarez@unsl.edu.ar} (N. Juarez), \texttt{paneme@unsl.edu.ar} (P. Neme) and \texttt{joviedo@unsl.edu.ar} (J. Oviedo).} \and  
Noelia Juarez\samethanks[2] \and Pablo Neme\samethanks[2] \and Jorge Oviedo\samethanks[2]}

\date{\today}
\maketitle

\begin{abstract}


In a many-to-one matching model with responsive preferences in which indifferences are allowed, we study  three  notions of core, three notions of stability, and their relationships.  We show that (i) the core contains the stable set, (ii) the strong core coincides with the strongly stable set, and (iii) the super core coincides with the super stable set. We also show how the core  and the strong core in markets with indifferences relate to the stable matchings  of their associated  tie-breaking strict markets.

\bigskip

\bigskip


\end{abstract}

\section{Introduction}


The core  is one of the most important solution concepts in cooperative games.  Two-sided matching models can be seen as particular cooperative games. In these models, a matching is ``dominated'' if there is a coalition of agents that prefer another matching in which each agent of the coalition is matched to agent(s) within the coalition
; and a matching is in the ``core'' if it is undominated. However,  in these two-sided games, the most studied solution concept  is not the core but  that of stability. 
Unlike in models with strict preferences,
where there is a unique concept of stability satisfied by pairs of agents, in models with indifferences there are several concepts of stability. \cite{irving1994stable} introduces the notions of strongly stable matching and super stable matching.
A matching is ``blocked'' if it can be dominated via  a coalition made of a single agent or a pair of agents (one from each side of the market); and a matching is ``stable'' if it is not blocked.  Notice that in the definition of domination (and therefore in the definition of core) every coalition plays a potential role, whereas in the definition of blocking (and therefore in the definition of stability) only singletons and pairs of agents are involved. In one-to-one matching model with strict preferences, nothing is lost by ignoring coalitions other than singletons and pairs: the core coincides with the set of stable matchings \citep[see][]{roth1985common}.  But in the more general many-to-one model, a matching can be in the core and  not be stable. In our first theorem we show that this is also the case when indifferences are allowed.

In a many-to-one matching market with responsive preferences and indifferences, we study  two weaker versions of domination (with their corresponding stronger notions of core) and their associated weaker notions of blocking (with their corresponding stronger notions of stability) and investigate their relationships.   A matching is ``weakly dominated'' if (i) there is a coalition of agents that are at least as well-off in another matching, (ii) one of them prefers this new matching, and  (iii) agents of the coalition are matched within the coalition in this new matching; and a matching is in the  ``strong core'' if it is not weakly dominated. A matching is ``weakly blocked'' if it can be weakly dominated via  a coalition made of a single agent or a pair of agents (one from each side of the market); and a matching is ``strongly stable'' if it is not weakly blocked. An even weaker notion of domination, called ``super weak domination'' can be defined by losing condition (ii) in the definition of weak domination. The associated definitions of ``super core'', ``super weak blocking'' and ``super stable matching'' are straightforward. 
In this paper, we show that (i) the strong core coincides with the set of strongly stable matchings, and (ii) the super core coincides with the set of super stable matchings.

\cite{gale1962college} show that at least one stable matching always exists, even when agents may have indifferences in their preferences.
Usually, the procedure to compute a stable matching consist of breaking ties and then applying
Gale and Shapley’s deferred acceptance algorithm. How these indifferences may be
ordered has both strategic and welfare consequences \citep[see][]{erdil2017two,abdulkadirouglu2009strategy}. On the other hand, strongly stable matchings and super stable
matchings may not exist. In the one-to-one model,
\cite{irving1994stable} presents algorithms for determining whether a strongly stable matching and/or a super
stable matching exist. In each case, if such a matching exists, the appropriate
algorithm constructs one. \cite{manlove2002structure} shows that the set of strongly stable matchings
forms a distributive lattice.  \cite{kunysz2016characterisation} present a polynomial-time algorithm
for the generation of all strongly stable matchings. They also give  an alternative prove that the set of
strongly stable matching forms a distributive lattice. In the many-to-one model, \cite{irving2000hospitals} present the first linear-time algorithm to compute super stable matchings when it exists. 
 Linear programming approaches to the study of stable and strongly stable matchings in models with indifferences are also available
\citep[see][]{kwanashie2014integer,juarez2021marriage,kunysz2018algorithm}.
Besides these three notions of stability, there are other solution concepts for models with indifferences. \cite{sotomayor2011pareto} proposes Pareto-stability as a solution concept for the one-to-one model with indifferences and the roommate model. A matching is Pareto-stable if it is  stable and Pareto optimal.  She shows that every strongly stable matching is a Pareto-stable matching and exhibits examples  of markets in which there is a Pareto-stable matching that is not strongly stable and there is a stable matching that is not Pareto-stable.

The rest of the paper is organized as follows.  Preliminaries are presented in Section \ref{Preliminaries}, where also 
the three notions of stability and the three notions of core are introduced. Our results are shown is Section \ref{Results}. After that, some conclusions are gathered in Section \ref{Concludings}.   

\section{Preliminaries}\label{Preliminaries}

Let $F$ and $W$ denote disjoint finite sets of \emph{firms} and \emph{workers}, respectively. 
Each firm $f \in F$ has a \emph{preference} relation $R_f$ over $2^W$  that is complete and transitive (i.e. a weak order). Each  worker has a \emph{preference} relation $R_w$ that is a weak order over $F \cup \{\emptyset\},$ where $\emptyset$ represents the prospect of being unemployed. Let  $R=(R_a)_{a \in F \cup W}$ denote the \emph{profile of preferences} of the agents. For each $a \in F \cup W,$ let $P_a$ and $I_a$ denote the antisymmetric and symmetric part of $R_a,$ respectively. Throughout, we assume that there is no  firm $f \in F$ such that $W' I_f \emptyset$ for some $W' \subseteq W$ and  there is no worker $w \in W$ such that $f I_w \emptyset$ for some $f \in F.$ We call this the \emph{no indifference to the empty set} assumption.\footnote{This assumption is commonly
used in the literature; see Erdil and Ergin (2008, 2017).}  For each $f \in F,$ let $q_f \geq 1$ denote the number of positions that firm $f$ has (called $f$'s \emph{quota}), and let $q=(q_f)_{f \in F}$ be the \emph{profile of quotas}. 
We assume that each firm  gives its ranking of workers individually, and orders subsets of workers in a \textit{responsive} manner. That is to say, adding ``good" workers to a set leads to a better set, whereas adding ``bad" workers to a set leads to a worse set. In addition, for any two subsets that differ in only one worker, the firm prefers the subset containing the most preferred worker. Formally,
\begin{definition}\label{def responsive}
Given $q_f,$  preference relation $R_f$ over $2^{W}$ is \textbf{responsive} if it satisfies the following conditions:\footnote{
Notice that when $w'=\emptyset,$ by \emph{no indifference to the empty set,} Conditions (i) and (ii) in Definition \ref{def responsive} imply that for each $T\subseteq{W}$ such that $\vert{T}\vert<q_{f}$ and each $w\in W \setminus T$,   $$T\cup\{w\}R_{f}T \mbox{ if and only if } wP_{f}\emptyset.$$
}
\begin{enumerate}[(i)]
\item for each $T\subseteq{W}$ such that $\vert{T}\vert>q_{f}$,   $\emptyset P_{f}T$.
\item for each $T\subseteq{W}$ such that $\vert T  \vert \leq q_{f},$  each $w'\in T\cup \{ \emptyset\}$, and each $w\in W \setminus T$,  
$$(T \setminus \{w'\})\cup \{w\}R_{f} T \mbox{ if and only if } wR_{f}w'.\footnote{Condition (ii) in Definition \ref{def responsive} is equivalent to the following: \emph{
\begin{itemize}
\item[(ii)'] for each $T\subseteq{W}$ such that $\vert T  \vert \leq q_{f},$  each $w'\in T\cup \{ \emptyset\}$, and each $w\in W \setminus T$, 
\begin{itemize}
\item[(ii.a)] $(T \setminus \{w'\})\cup \{w\}P_{f} T \mbox{ if and only if } wP_{f}w',$ and
\item[(ii.b)]$(T \setminus \{w'\})\cup \{w\}I_{f} T \mbox{ if and only if } wI_{f}w'.$
\end{itemize}
\end{itemize}
} }$$
\end{enumerate}

\end{definition}

A \emph{many-to-one (matching) market} is denoted by 
 $(F,W,R,q).$ When $q_f=1$ for each $f\in F$, the  market will be called a \emph{one-to-one (matching) market}.
A \textit{matching} $\mu$ is a mapping from  $F \cup W$ into $2^{F \cup W}$ 
 such that, for each $w\in{W}$ and each $f\in{F}$:
\begin{enumerate}[(i)]
\item  $\mu(w)\subseteq{F}$ and $\vert\mu(w)\vert \leq 1,$
\item $\mu(f)\subseteq W$ and $\vert\mu(f)\vert\leq q_{f},$
\item $\mu(w)=\{f\}$ if and only if $w\in\mu(f)$.
\end{enumerate}
Furthermore, the set of all matchings is denoted by $\mathcal{M}.$ 
Since workers are assigned to at most one firm, usually we will omit the curly brackets. For instance, instead of condition (iii) we will write: ``$\mu(w)=f$ if and only if $w\in \mu(f)$''.
 Throughout the paper we identify the market $(F,W,R,q)$ with its corresponding preference profile $R.$ 

\subsection{Stability, strong stability and super stability}

In matching models, stability is considered the main property to be satisfied by any matching. Unlike in models with strict preferences,
where there is a unique concept of stability satisfied by pairs of agents, in models with indifferences there are several concepts of stability.
In the one-to-one model, a matching is stable if each agent is matched to an acceptable partner, and
there is no firm-worker pair such that they are not matched together and
strictly prefer each other to their current partners. \cite{irving1994stable}  formulates
two other possible definitions of stability for the one-to-one model with indifferences. A matching is strongly stable if each agent is matched to an
acceptable partner, and there is no firm-worker pair such that they are not
matched together, one of them strictly prefers the other one to their current partner, and the other weakly prefers the other one to their current partner.
A matching is super stable if each agent is matched to an acceptable partner,
and there is no firm-worker pair such that they are not matched together
and weakly prefer each other to their current partners. 
 
 Now we present three types of blocking pairs, in order to formalize the different definitions of stability for the many-to-one setting.

\begin{definition}
Let $\mu $ be a matching. We say that the pair $(f,w)\in F \times W$ with $w\notin \mu(f)$ 

\begin{enumerate}[(i)]
\item \textbf{blocks} $\boldsymbol{\mu}$ whenever 
\begin{enumerate}
\item[(i.a)] $\left\vert \mu (f)\right\vert =q_{f},$ and there is $w^{\prime }$ $\in \mu (f)$ such that $wP_{f}w^{\prime }$ and $fP_{w}\mu (w)$; or

\item[(i.b)] $\left\vert\mu (f)\right\vert <q_{f},$  $wP_{f}\emptyset,$ and $fP_{w}\mu (w)$.

\end{enumerate}

\item \textbf{weakly blocks} $\boldsymbol{\mu}$ whenever
\begin{enumerate}
\item[(ii.a)] $\left\vert \mu (f)\right\vert =q_{f},$ and there is $w^{\prime }$ $\in \mu (f)$ such that either  $wR_{f}w^{\prime }$ and $fP_{w}\mu (w)$, or $wP_{f}w^{\prime }$ and $fR_{w}\mu (w)$; or 

\item[(ii.b)] $\left\vert\mu (f)\right\vert <q_{f},$  $wP_{f}\emptyset,$ and $fR_{w}\mu (w)$.

\end{enumerate}

\item \textbf{super weakly blocks} $\boldsymbol{\mu}$  whenever 

\begin{enumerate}

\item[(iii.a)] $\left\vert \mu (f)\right\vert =q_{f},$ and there is $w^{\prime } \in \mu (f)$ such that $wR_{f}w^{\prime }$ and $fR_{w}\mu (w)$; or 

\item[(iii.b)] $\left\vert\mu (f)\right\vert <q_{f},$  $wP_{f}\emptyset,$ and $fR_{w}\mu (w).$

\end{enumerate}

\end{enumerate}
\end{definition}

Now we are in a position to define the three notions of stability that we use throughout the paper. First, we present a common condition that all stability notion must satisfy. We say that a matching $\mu$ is  \textbf{individually rational} if it is not blocked by any individual agent, i.e.,  (i)  for each $w\in W$, $\mu
(w)P_{w}\emptyset $; and (ii) for each $f\in F$, $wP_{f}\emptyset $ for each $w \in \mu(f).$

\begin{definition} We say that an individually rational  matching  is
\begin{itemize}
\item[(i)] \textbf{stable} if it has no blocking pair.
\item[(ii)] \textbf{strongly stable} if it has no weakly blocking pair.
\item[(iii)] \textbf{super stable} if it has no super weakly blocking pair.
\end{itemize}
\noindent Denote  by $\mathcal{S}$, $\mathcal{SS}$ and $\mathcal{SSS}$ the set of all stable, strongly stable, and super stable matchings, respectively.
\end{definition}

The existence of stable matchings in models with indifferences
is guaranteed: by breaking ties arbitrarily, we obtain a model with strict preferences. 
A stable matching in this strict model is a stable matching in the original model with indifferences. \cite{gale1962college} show that each market with strict preferences has at least one stable matching and compute it by the deferred acceptance algorithm. On the other hand, strongly stable matchings and super stable matchings could not exist. \cite{irving1994stable} presents algorithms for determining whether a strongly stable matching and/or a super
stable matching exist. In each case, if such a matching exists, the appropriate
algorithm constructs one.

The following example shows that the super stable matching set can be a proper subset of the strongly stable matching set and, in turn, the strongly stable matching set can be a proper subset of the stable matchings set.

\begin{example}\label{ejemplo}
Let $(F,W,R,q)$ be a matching market where $F=\{f_1,f_2,f_3\}, W=\{w_1,w_2,w_3,w_4\},$  $q_{f_{1}}=2$, $q_{f_{i}}=1$ for $i=2,3 $ and the preference profile is given by:\footnote{``$R_{f_1}: w_1, w_4, [w_2, w_3]$'' indicates that $w_1 P_{f_1} w_4 P_{f_1} w_2 I_{f_1} w_3.$}
\medskip

\begin{center}\noindent\begin{tabular}{l}
$ R_{f_1}: w_1,w_4, \left[ w_2,w_3\right]  $\\
$ R_{f_2}: \left[ w_2,w_3\right] ,w_1,w_4 $\\
$ R_{f_3}:\left[ w_2,w_3\right],w_4,w_1  $\\
\end{tabular}~~~~~~~~~~~~~~~~~~~
\noindent\begin{tabular}{l}
$ R_{w_1}:f_3,f_1,f_2 $\\
$ R_{w_2}: f_1,\left[ f_2,f_3\right]   $\\
$ R_{w_3}:f_1,\left[ f_2,f_3\right]  $\\
$ R_{w_4}:f_2,f_3,f_1  $\\
\end{tabular}\medskip
\end{center}
\noindent Consider the following stable matchings:
\begin{center}
 $
\mu_1=\begin{pmatrix}
f_1 & f_2 &f_3 \\
w_2 w_3& w_4 &w_1 \\
\end{pmatrix}
,~$
$
\mu_2=\begin{pmatrix}
f_1 & f_2 &f_3 \\
w_1w_4& w_2 &w_3 \\
\end{pmatrix}
,$ and 
$
\mu_3=\begin{pmatrix}
f_1 & f_2 &f_3 \\
w_1w_2& w_4 &w_3 \\
\end{pmatrix}.
$
\end{center}

\noindent Notice that $ \mu_1 $ matches each worker to her most preferred firm. Since each worker does not have indifference in the top of her preference, there are no super  weakly blocking pairs for $\mu_1 $. Therefore, $\mu_1\in \mathcal{SSS}  $. 
Since $(f_2,w_3)$ is a super weakly blocking pair for  $\mu_2$, $\mu_2 \notin \mathcal{SSS}.$  It can be checked that $\mu_2$ does not have any strongly blocking pair, thus $\mu_2\in \mathcal{SS}  $. 
Furthermore, since $ (f_1,w_3)$ is a  weakly blocking pair for  $\mu_3$,   $\mu_3 \notin \mathcal{SS}  $.
\hfill $\Diamond$

\end{example}



\subsection{Core, strong core, and super core}

A cooperative game consists of a set of players; a set of feasible outcomes; preferences of the players over outcomes; and the ``rules of the game'', that specify which coalitions of players are empowered to enforce which outcomes \citep[see Chapter 3 in][]{roth1992two}. Every many-to-one matching market with indifferences $(F,W,R,q)$ induces a cooperative game as follows: the set of players is, of course, $F \cup W;$ the feasible outcomes are equal to the set of all possible matchings $\mathcal{M};$ the preferences over outcomes (matchings) are induced from the profile $R$ of preferences over (subsets) of agents in the straightforward way; and the rules of game specify that a coalition $C \subseteq F \cup W$ is empowered to enforce matching $\mu' \in \mathcal{M}$ if and only if $\mu'(C)=C.$  This implies that every worker in coalition $C$ is assigned to some
firm in $C$ by $\mu'$ and that every firm in coalition $C$ is assigned to some subset of workers in $%
C$ by $\mu'.$ 

The rules of the game and the preferences of the agents  allow us to define three domination relations on the outcomes of the game:

\begin{definition}
Let $\mu, \mu' \in \mathcal{M}$ and let $C\subseteq F\cup W$
be a non-empty coalition such that $\mu'(C)=C.$ 

\begin{itemize}
\item[(i)] $\boldsymbol{\mu'$ \textbf{dominates $\mu$ via} $C}$ if $\mu' (c)P_{c}\mu(c)$ for each $c\in C.$

\item[(ii)] $\boldsymbol{\mu'$ \textbf{weakly  dominates $\mu$ via}  $C}$ if $\mu'(c)R_{c}\mu(c)$  for each $c \in C,$ and there is 
$c'\in C$ such that $\mu'(c')P_{c'}\mu(c').$

\item[(iii)] $\boldsymbol{\mu'$ \textbf{super  weakly dominates $\mu$ via} $C}$ if there is $c\in C$ such that $\mu(c)\neq\mu'(c)$, and  $\mu'(c')R_{c'}\mu(c')$  for each $c' \in C.$ 
\end{itemize}
\end{definition}

\begin{definition} The set of all matchings that are not

\begin{itemize}
\item[(i)] dominated is the \textbf{core}. 

\item[(ii)] weakly dominated is the \textbf{strong core}.

\item[(iii)] super weakly dominated is the \textbf{super  core}.
 
\end{itemize}

\noindent Denote by $\mathcal{C},$ $\mathcal{C}_S,$ and $\mathcal{C}_{SS}$ the core, the strong core, and the super  core, respectively. 

\end{definition}

\noindent By definition, the core includes the strong core and the strong core includes the super core.


\section{Results}\label{Results}
In this section, unless otherwise specified, all the results apply to the many-to-one model with indifferences.
In the one-to-one setting with strict preferences, the core coincides with the strong core  \citep{roth1985common}. Next, we show that this result holds even when indifferences are allowed.

\begin{proposition}\label{propo 1 a 1 core = estables}
For any one-to-one market with indifferences, $\mathcal{C}=\mathcal{S}$.
\end{proposition}
\begin{proof}
First, we prove that $\mathcal{C}\subseteq \mathcal{S}.$ Assume that $\mu \notin \mathcal{S}$. If $\mu $ is not an
individually rational matching, then there is $c\in F\cup W$ such that $\emptyset P_{c}\mu (c).$ Let $\mu'$ be a matching such that $\mu'(c)=\emptyset$. Thus, $\mu' $ dominates $\mu$ via $C=\{c\}.$ Now, assume that there is a blocking pair $(f,w)$ of $\mu$. That is,  $wP_{f}\mu(f) $ and $fP_{w}\mu (w).$ Let matching $\mu'$ be such that $\mu'(w)=f$. Let $C=\{f,w\}$. Obviously, $\mu'(C)=C$ and it is straightforward that $\mu'$ dominates $\mu$ via $C.$ Therefore, $\mu \notin \mathcal{C}.$
Next, we prove that $\mathcal{S}\subseteq \mathcal{C}.$ Assume that $\mu \in \mathcal{S}$ and $\mu \notin \mathcal{C}.$ Then, there are a matching $\mu'$ and a non-empty coalition $C \subseteq F\cup W$ such that $\mu' $ dominates $\mu$ via $C$. 
Take any  $c\in C.$ If $\mu'(c)=\emptyset,$ since $\mu'$  dominates  $\mu,$ we have $\emptyset=\mu'(c)P_{c}\mu (c),$ contradicting 
the individual rationality of $\mu.$ Hence, $\mu'(c)\neq \emptyset$ for each  $c\in C,$ implying $W\cap C\neq \emptyset.$  Let $w\in W\cap C$ and let $f=\mu'(w).$ Since $\mu'(C)=C$, $f\in C$. Since $\mu'$ dominates $\mu $, it follows that $f=\mu'(w)P_{w}\mu (w)$  and $w=\mu'(f)P_{f}\mu (f)$.
Thus,  $(f,w)$ blocks $\mu,$ contradicting the stability of $\mu.$ Therefore,  $\mu \in \mathcal{C}.$
\end{proof}

In the many-to-one setting, the core is no longer equal to the stable matching set, only one inclusion holds:

\begin{theorem}\label{stable in core}
$\mathcal{S} \subseteq \mathcal{C}$.
\end{theorem}
\begin{proof}
Let $\mu \in \mathcal{S}$ and assume $\mu \notin \mathcal{C}.$ Then, there are a non-empty coalition $C \subseteq F \cup W$ and a matching $\mu' \in \mathcal{M}$ such that  $\mu'(C)=C$ and $\mu'(c) P_c \mu(c)$ for each $c \in C.$  First, we claim that $C\cap F \neq \emptyset.$ Otherwise,  $C\cap F = \emptyset$ implies, for each $w \in C,$ that $\mu'(w)=\emptyset$ (since $\mu'(C)=C$), and therefore $\emptyset=\mu'(w) P_w \mu(w),$ contradicting the individual rationality of $\mu.$ So take any $f\in C$. Since $\mu'$ dominated $\mu$ via $C$, $\mu'(f)P_f \mu(f).$ Now we claim that $\mu'(f)\setminus \mu(f)\neq \emptyset.$  Otherwise, $\mu'(f)\setminus \mu(f)=\emptyset$, $\mu'(f)P_f \mu(f)$ imply $\mu'(f)\subset\mu(f).$ Let $\mu(f)\setminus \mu'(f):= \{w_1,\ldots,w_\ell\}$. By individual rationality of $\mu$, $w_iP_f\emptyset$ for each  $i=1,\ldots,\ell.$ By responsiveness of $R_f$, 
$$
\mu(f)=\mu'(f) \cup\{w_1,\ldots,w_\ell\} P_f \mu'(f) \cup\{w_1,\ldots,w_{\ell-1}\}P_f \cdots P_f \mu'(f) \cup\{w_1\} P_f \mu'(f).
$$
Thus $\mu(f) P_f \mu'(f),$ contradicting $\mu'(f)P_f \mu(f)$. Therefore, $\mu'(f)\setminus \mu(f)\neq \emptyset.$ Define 
$$
\mu'(f)\setminus \mu(f):= \{w'_1,\ldots,w'_k\}
$$
where subscripts are chosen so that $w'_i R_f w'_{i+1}$ for each $i=1\ldots,k-1,$ and  
$$
\mu(f)\setminus \mu'(f):= \{w_1,\ldots,w_\ell\}
$$
where subscripts are chosen so that $w_i R_f w_{i+1}$ for each $i=1\ldots,\ell-1.$
Now, there are two cases to consider:
\begin{itemize}
\item[$\boldsymbol{1}.$]$\boldsymbol{|\mu'(f)|\leq|\mu(f)|}$. Then, $|\mu'(f)\setminus \mu(f)|\leq|\mu(f)\setminus \mu'(f)|.$ First, we claim that  
\begin{equation}\label{chorizo 0 lemma 1}
w'_1 P_f w_\ell.
\end{equation} Otherwise, $$
w_1 R_f  \cdots R_fw_\ell R_f w'_1 R_f  \cdots R_f w'_k$$ 
and, by responsiveness of $R_f$ and the individual rationality of $\mu$, since $\ell \geq k$ 
\begin{equation}\label{chorizo 1 lemma 1}
\mu(f) P_f \mu(f) \setminus \{w_\ell \} P_f \cdots P_f \mu(f) \setminus \{w_{k+1},\ldots,w_\ell \}.
\end{equation}
Define $\widetilde{W}:=\mu'(f)\cap \mu(f).$ Notice that $\mu(f) \setminus \{w_{k+1},\ldots,w_\ell \} = \widetilde{W}\cup \{w_1,\ldots,w_k\}.$ Again, by responsiveness of $R_f$, 
\begin{equation}\label{chorizo 2 lemma 1}
\widetilde{W}\cup \{w_1,\ldots,w_k\} R_f \widetilde{W}\cup \{w_1,\ldots,w'_k\} R_f \cdots R_f \widetilde{W}\cup \{w'_1\ldots,w'_k\}.
\end{equation}
Notice that $\widetilde{W}\cup \{w'_1,\ldots,w'_k\}= \mu'(f).$ Thus, by \eqref{chorizo 1 lemma 1} and \eqref{chorizo 2 lemma 1}, $\mu(f)R_f \mu'(f)$. This contradicts that  $\mu'(f) P_f \mu(f)$ and, therefore, \eqref{chorizo 0 lemma 1} holds and the claim is proven. Since $f\in C$ and $\mu'(C)=C$, $w'_1\in C$. As $\mu'$  dominates $\mu$ via $C$, 
\begin{equation}\label{chorizo 3 lemma 1}
f=\mu'(w'_1) P_{w'_1} \mu(w'_1).
\end{equation} 
If $|\mu(f)|=q_f$, by \eqref{chorizo 0 lemma 1} and \eqref{chorizo 3 lemma 1},   $(f,w'_1)$ is a blocking pair for $\mu$.  If $|\mu(f)|<q_f$, by \eqref{chorizo 0 lemma 1}  we have $w'_1 P_f w_\ell$, and by the individual rationality of $\mu$, $w_\ell P_f \emptyset.$ Hence, $w'_1 P_f \emptyset$ and, therefore,   together with \eqref{chorizo 3 lemma 1},   $(f,w'_1)$ is a  blocking pair for $\mu$.

\item[$\boldsymbol{2}.$] $\boldsymbol{|\mu'(f)|>|\mu(f)|}$. First, we claim that  
\begin{equation}\label{chorizo 4 lemma 1}
w'_1 P_f \emptyset.
\end{equation} Otherwise, $$
\emptyset P_f w'_1 R_f  \cdots R_f w'_k.$$ 
Define $\widetilde{W}:=\mu'(f)\cap \mu(f).$ By responsiveness of $R_f$, since $\ell < k$ 
\begin{equation}\label{chorizo 5 lemma 1}
\mu(f) P_f (\mu(f) \setminus \{w_1 \})\cup\{w_1'\} P_f \cdots P_f \widetilde{W} \cup \{w_{1}',\ldots,w_\ell' \}.
\end{equation}
 
 Again, by responsiveness of $R_f$, 
\begin{equation}\label{chorizo 6 lemma 1}
\widetilde{W}\cup \{w_1',\ldots,w_\ell'\} P_f \widetilde{W}\cup \{w_1',\ldots,w'_\ell\ , w_{\ell+1}'\} P_f \cdots P_f \widetilde{W}\cup \{w'_1\ldots,w'_k\}.
\end{equation}
Notice that $\widetilde{W}\cup \{w'_1,\ldots,w'_k\}= \mu'(f).$ Thus, by \eqref{chorizo 5 lemma 1} and \eqref{chorizo 6 lemma 1}, $\mu(f) P_f \mu'(f)$. This contradicts that  $\mu'$ dominates $\mu$ via $C$ and, therefore, \eqref{chorizo 4 lemma 1} holds and the claim is proven. 

Notice that since $\mu'$ is a matching, $q_f \geq |\mu'(f)|>|\mu(f)|$. Since $f\in C$ and $\mu'(C)=C$, $w'_1\in C$. As $\mu'$ dominates $\mu$ via $C$, \eqref{chorizo 3 lemma 1} holds, and together with  \eqref{chorizo 4 lemma 1} imply that  $(f,w'_1)$ is a blocking pair for $\mu$. 
\end{itemize}
 
Given that in each case we find a blocking pair for $\mu$, we  contradict the fact   that $\mu \in\mathcal{S}.$ Therefore, $\mu \in\mathcal{C}.$ 
\end{proof}

Since the stable matching set is always non-empty, the following holds:

\begin{remark}
$\mathcal{C}\neq\emptyset.$
\end{remark}

Next, using Example \ref{ejemplo},  we show that the inclusion stated in Theorem \ref{stable in core} can be a proper inclusion.

\noindent \textbf{Example 1 (continued)}
\emph{Consider matching
$
\mu_4=\begin{pmatrix}
f_1 & f_2 &f_3 \\
w_2w_3& w_1 &w_4 \\
\end{pmatrix}
$.
The pair $ (f_1,w_1) $ blocks  $\mu_4$ and, therefore,  $\mu_4 \notin \mathcal{S}  $. However, we will show that $ \mu_4 \in \mathcal{C} $. Assume that this is not the case. Then, there are  a coalition of agents $ C $ and a matching $ \mu' $ such that $ \mu'(C)=C $ and $ \mu' $ dominates $ \mu_4 $ via $ C $. Since $ w_2 $ and $ w_3 $ are matched to their most preferred firm in $\mu_4$, then $ w_2 $ and $ w_3 $  cannot be part of coalition $ C $.  Assume that $ w_4 \in C $. Since $\mu_4( w_4)=f_3$  and  $f_2$ is the only firm preferred by $w_4$ to $f_3$,  $ f_2 $ must belong to coalition $C$ and $\mu'(w_4)=f_2.$  However, $  \mu_4( f_2)=w_1P_{f_2}w_4=\mu'(f_2)$ contradincting that $ \mu' $ dominates $ \mu_4 $ via $ C $.
Lastly, assume that $ w_1\in C $. Thus, either $f_1$ or $f_3$ belong to $C$. If $f_1\in C$ and $ \mu'(w_1)=f_1$,  in order that $ \mu' $ dominates $ \mu_4 $ via $ C $ we need that $|\mu'(f_1)|=2$. Then, some other worker besides $w_1$ must belong to $C$, contradicting the previous arguments. If $f_3\in C$ and $ \mu'(w_1)=f_3$, then $\mu_4(f_3)=w_4 P_{f_3}w_1=\mu'(f_3)$. This contradicts that $ \mu' $ dominates $ \mu_4 $ via $ C $. Therefore,  $\mu_4$ is undominated implying that $\mu_4 \in \mathcal{C} \setminus \mathcal{S}. $}
\hfill $\Diamond$

\vspace{15 pt}

In a many-to-one setting with strict and responsive  preferences, the strong core coincides with the set of stable matchings   \citep[][]{roth1985common}. However, if we allow indifferences, the following result shows that  the strong core coincides with the   set of strongly stable matchings.

\begin{theorem}\label{strong core = strongly stable}
$\mathcal{C}_{S}=\mathcal{SS}.$
\end{theorem}
\begin{proof}
First, we prove that $\mathcal{C}_{S} \subseteq \mathcal{SS}.$ Assume that $\mu \notin \mathcal{SS}$. We want to see that $\mu \notin
\mathcal{C}_{S}.$  There are two cases to consider:  

\begin{itemize}

\item[$\boldsymbol{1}$.] \textbf{$\boldsymbol{\mu}$ is not individually rational.} Thus, there are  $f\in F$ and $w\in \mu(f)$ such that either $\emptyset P_{w}f $ or $\emptyset P_f w .$
If $\emptyset P_{w}f$, let $\mu ^{\prime }$ be a matching such that $\mu
^{\prime }(w)=\emptyset $. Thus, $\mu'$ weakly dominates  $\mu $ via  coalition $\left\{ w\right\} .$ Therefore $\mu \notin
\mathcal{C}_{S}.$ 
If there is $w\in \mu(f)$ such that  $\emptyset P_f w ,$ let $\mu'$ be a matching such that $\mu'(f)=\mu(f) \setminus \{w\}$. Thus, $\mu'$ weakly dominates $\mu$ via  coalition $C=\{f\} \cup \left(\mu(f)\setminus \{ w\}\right).$ Therefore $\mu \notin
\mathcal{C}_{S}.$

\item[$\boldsymbol{2}$.] \textbf{$\boldsymbol{\mu}$ is  individually rational.} Then, since $\mu \notin \mathcal{SS},$ there is a weakly blocking pair  $(f,w')$ of $\mu.$ There are two subcases to consider:

\begin{enumerate}

\item[$\boldsymbol{2.1}$.] \textbf{$\boldsymbol{|\mu(f)|=q_f}.$} 
Since $(f,w')$ is a weakly blocking pair for $\mu$,  there is $\overline{w}$ $\in \mu (f)$ such that 
either  $w'P_{f}\overline{w}\text{ and }fR_{w'}\mu (w')$  or $w'R_{f}\overline{w}\text{ and } fP_{w'}\mu (w').$
  Consider  coalition  $C=\{ f,w'\} \cup \left(\mu
(f)\setminus \{\overline{w}\} \right)$.  Let matching $\mu'$ be
such that $\mu'(w')=f$ and $\mu'(w)=f$ for each $w\in \mu(f) \setminus \{\overline{w}\}.$ This implies that  $\mu'(f)=\left( \mu (f)\setminus
\{\overline{w}\}\right) \cup \{w'\}$. By definition of $\mu'$, $\mu'(C)=C.$  We claim that $\mu ^{\prime }$ dominates $%
\mu $ via $C.$
Since $w'P_f \overline{w}$, by responsiveness, 
\begin{equation}\label{ecu 01 teor 1}
\mu'(f)=\left( \mu (f)\setminus\{\overline{w}\}\right) \cup \{w'\}P_f \mu(f).
\end{equation}
 By definition of $\mu'$, $\mu'(w)=\mu(w)$ and, therefore,  $\mu'(w)I_w \mu(w)$ for each $w\in C\setminus \{w'\}$. Finally, since $(f,w')$ is a weakly blocking pair for $\mu$, 
 \begin{equation}\label{ecu 02 teor 1}
 f=\mu'(w')P_{w'}\mu(w').
  \end{equation} Hence, $\mu'$ weakly dominates $\mu$ via $C$.

\item[$\boldsymbol{2.2}$.] \textbf{$\boldsymbol{|\mu(f)|<q_f}.$}  Consider coalition $C=\{ f,w'\} \cup \mu (f) $. Let matching $\mu'$ be
such that $\mu'(w')=f$ and $\mu'(w)=f$ for each $w\in \mu(f).$ This implies that  $\mu'(f)=\mu (f) \cup \{w'\}$. By definition of $\mu'$, $\mu'(C)=C.$   We claim that $\mu'$ weakly dominates $\mu $ via $C.$ First consider $w\in \mu (f)\setminus \{w'\}$. By definition of $\mu',$ $\mu '(w)=\mu (w)$
 and then $\mu'(w )R_{w }\mu (w)$. Since $(f,w')$ is a
weakly blocking pair for $\mu$, we have that $fR_{w'}\mu (w')$ and $w'P_{f}\emptyset $.  
Definition of $\mu'$, $w'P_{f}\emptyset$, and responsiveness imply  $\mu'(f)P_{f}\mu (f)$. Furthermore, $\mu'(w')=f$ and $fR_{w'}\mu (w')$ imply $\mu'(w')R_{w'}\mu (w').$ Thus, $\mu'$ weakly dominates $\mu $ via $C$. Therefore $\mu \notin \mathcal{C}_{S}.$ 
\end{enumerate}

\end{itemize}

\noindent Next, we prove that $\mathcal{SS}\subseteq \mathcal{C}_{S}.$ Assume that $\mu \in
\mathcal{SS} $ and $\mu \notin\mathcal{C}_S.$ Thus, there are a non-empty coalition $C$ and a matching $\mu'$ such that $\mu'$ weakly dominates $\mu $ via  $C.$  There are two cases to consider:
\begin{enumerate}

\item[$\boldsymbol{1}.$] \textbf{There is  $\boldsymbol{f\in C$ such that $\mu'(f)P_f \mu(f)}$}.
Now, following similar reasonings to the ones in the proof of Theorem \ref{stable in core}, and replacing \eqref{chorizo 3 lemma 1}  by
\[
f=\mu'(w'_1) R_{w'_1} \mu(w'_1), \tag{\ref{chorizo 3 lemma 1}'}
\]
we can prove that $(f,w'_1)$ is a weakly blocking pair for $\mu$.

\item[$\boldsymbol{2}.$] \textbf{There is $\boldsymbol{\widehat{w}\in C$ such that $\mu'(\widehat{w})P_{\widehat{w}}\mu(\widehat{w}) }$.} Let $f=\mu'(\widehat{w})$. Thus, $f\in C$. Since $\mu'$ weakly dominates $\mu$ via $C$, $\mu'(f)R_{f}\mu(f) .$ Assume w.l.o.g. that $wP_f \emptyset$ for each $w\in \mu'(f)$.\footnote{Otherwise, if $w'\in \mu'(f)$ is such that $\emptyset P_f w'$ we have, by responsiveness,  $\mu'(f)\setminus \{w'\}P_f \mu'(f)$. Let $\mu''(f)=\mu'(f)\setminus \{w'\}$ and $\mu''(a)=\mu'(a)$ for each $a\in C \setminus \{f,w'\}.$ Then, $\mu''$ dominates $\mu$ via $C\setminus\{w'\}.$} There are two subcases to consider:

\begin{itemize}
\item[$\boldsymbol{2.1}.$]$\boldsymbol{\mu'(f)P_{f}\mu(f) }$. Following the arguments of Case 1, the pair $(f,\widehat{w})$ is a weakly blocking pair for $\mu.$

\item[$\boldsymbol{2.2}.$]$\boldsymbol{\mu'(f)I_{f}\mu(f) }$.
First consider the case $|\mu'(f)|>|\mu(f)|$. Since $\mu'$ is a matching, $q_f \geq |\mu'(f)|>|\mu(f)|$. Also $\widehat{w}\in C$ such that $f=\mu'(\widehat{w})P_{\widehat{w}}\mu(\widehat{w})$ together with the fact that $\widehat{w}P_f \emptyset$ imply that $(f,\widehat{w})$ is a weakly blocking pair for $\mu.$

Now consider the case $|\mu'(f)|\leq|\mu(f)|$. Then, $|\mu'(f)\setminus \mu(f)|\leq|\mu(f)\setminus \mu'(f)|.$ By hypothesis, $\widehat{w}\in \mu'(f)\setminus \mu(f).$

Let $$\mathcal{X}=\{w\in \mu'(f)\setminus \mu(f):wR_f w_\ell\},$$ and
 $$\mathcal{Y}=\{w\in \mu'(f)\setminus \mu(f):f=\mu'(w) P_w \mu(w)\}.$$ 
 
Note that $\mathcal{Y}\neq \emptyset$, since $\widehat{w}\in \mathcal{Y}$.

First, consider  $\mathcal{X}\cap \mathcal{Y} \neq \emptyset$. Let $\overline{w} \in \mathcal{X}\cap \mathcal{Y}.$ Then, the pair  $(f,\overline{w})$ is a weakly blocking pair for $\mu.$

Second, consider $\mathcal{X}\cap \mathcal{Y}= \emptyset$. Then, $\mathcal{Y} \subseteq  \mathcal{X}^c$ and since $\mathcal{Y}\neq \emptyset, ~\mathcal{X}^c\neq \emptyset.$\footnote{By $\mathcal{X}^c$ we denote the complement of $\mathcal{X}.$}
We claim that 
\begin{equation}\label{ecu1 sub caso 2.2.2}
\text{there is  $\overline{w}\in \mathcal{X}$ such that $\overline{w}P_f w_\ell$.}
\end{equation}  Otherwise, $w_\ell R_f w$ for each $w\in \mathcal{X}$.  Moreover, by definition of $\mathcal{X}$, $w R_f w_\ell $ for each $w\in \mathcal{X}$. Thus, $w I_f w_\ell $ for each $w\in \mathcal{X}.$ Let $\mathcal{X}=\{w'_1,\ldots,w'_m\}.$ Notice that $m< k$. Hence, since $w I_f w_\ell $ for each $w\in \mathcal{X}$,  
\begin{equation}\label{ecu2 sub caso 2.2.2}
\mu(f)R_f \mu(f)\setminus \{w_1\} \cup \{w'_1\} R_f \cdots R_f  
\mu(f)\setminus \{w_1,\ldots,w_m\} \cup \{w'_1,\ldots,w'_m\}.
\end{equation}

Now, since $w_\ell P_f w $ for each $w\in \mathcal{X}^c$,  we have
\begin{equation}\label{ecu3 sub caso 2.2.2}\begin{split}
\mu(f)\setminus \{w_1,\ldots,w_m\} \cup \{w'_1,\ldots,w'_m\} P_f &\\ 
\mu(f)\setminus\{w_1,\ldots,w_m, w_{m+1}\} \cup \{&w'_1,\ldots,w'_m,w'_{m+1}\}P_f \cdots P_f \\
&\mu(f)\setminus\{w_1,\ldots,w_k\} \cup \{w'_1,\ldots,w'_k\}.
\end{split}
\end{equation}
Lastly, by the individual rationality of $\mu$, $wP_f \emptyset$ for each $w\in \mu$ 
\begin{equation}\label{ecu4 sub caso 2.2.2}
\begin{split}
\mu(f)\setminus \{w_1,\ldots,w_k\} \cup \{w'_1,\ldots,w'_k\}P_f &\\ 
\mu(f)\setminus \{w_1,\ldots,w_k,w_{k+1}\} \cup& \{w'_1,\ldots,w'_k\}P_f \cdots P_f \\
 \mu(f)\setminus \{w_1,&\ldots,w_\ell\} \cup \{w'_1,\ldots,w'_k\}=\mu'(f).
\end{split}
\end{equation}
By \eqref{ecu2 sub caso 2.2.2}, \eqref{ecu3 sub caso 2.2.2} and \eqref{ecu4 sub caso 2.2.2}, $\mu(f)P_f \mu'(f)$. This contradicts this subcase's hypothesis. Then, \eqref{ecu1 sub caso 2.2.2} holds. Since  $f\in C$ and  $\overline{w}\in \mu'(f)$, $\overline{w}\in C$. Then, $\mu'(\overline{w})R_{\overline{w}} \mu(\overline{w})$. This fact together with  \eqref{ecu1 sub caso 2.2.2} imply that $(f,\overline{w})$ is a weakly blocking pair for $\mu.$

\end{itemize}

\end{enumerate} 
Given that in each case we find a weakly blocking pair for $\mu$, we contradict the fact  that $\mu \in\mathcal{SS}.$ Therefore, $\mu \in\mathcal{C}_S.$ 
\end{proof}

Note, however, that the strong core may be empty. To see this, consider the following example due to \cite{roth1992two} (p. 167). There are two workers and one firm, with quota equal to one, that is acceptable to but indifferent between both workers. Then both individually rational matchings are weakly dominated (by one another) even though both are stable. This example, together with Theorem \ref{strong core = strongly stable} confirms the fact that, as we previously mentioned, the set $\mathcal{SS}$ may be empty. 

The previous theorem shows that nothing is lost by ignoring coalitions other than singletons and pairs when studying weak domination:  the strong core coincides with the set of strongly stable matchings. The following result shows that the same is true for super weak domination:

\begin{theorem} \label{super core = super stable}
$\mathcal{C}_{SS}=\mathcal{SSS}.$
\end{theorem}
\begin{proof}
First, we prove that $\mathcal{C}_{SS} \subseteq \mathcal{SSS}.$ Assume that $\mu \notin \mathcal{SSS}$. We want to see that $\mu \notin
\mathcal{C}_{SS}.$  There are two cases to consider:  

\begin{itemize}

\item[$\boldsymbol{1}$.] \textbf{$\boldsymbol{\mu}$ is not individually rational.}  To prove this case, we can construct  matchings $\mu'$ that super weakly dominates $\mu$ in an analogous way to Case 1 in the proof of Theorem \ref{strong core = strongly stable}.

\item[$\boldsymbol{2}$.] \textbf{$\boldsymbol{\mu}$ is  individually rational.} Then, since $\mu \notin \mathcal{SSS},$ there is a super weakly blocking pair  $(f,w')$ of $\mu.$ There are two subcases to consider:

\begin{enumerate}

\item[$\boldsymbol{2.1}$.] \textbf{$\boldsymbol{|\mu(f)|=q_f}.$} 
Since $(f,w')$ is a super weakly blocking pair for $\mu$,  there is $\overline{w}$ $\in \mu (f)$ such that 
$  w'R_{f}\overline{w}\text{ and }fR_{w'}\mu (w'). $
Now, following similar reasonings to the ones in Case 2.1 in the proof of Theorem \ref{strong core = strongly stable}, and replacing \eqref{ecu 01 teor 1} and 
\eqref{ecu 02 teor 1} by
\[\mu'(f)=\left( \mu (f)\setminus\{\overline{w}\}\right) \cup \{w'\}R_f \mu(f), \tag{\ref{ecu 01 teor 1}'}
\]
and 
 \[
 f=\mu'(w')R_{w'}\mu(w'), \tag{\ref{ecu 02 teor 1}'}
 \]
we can construct coalition $C$ and matching $\mu'$ such that $\mu'$ super weakly dominates $\mu$ via $C$.

\item[$\boldsymbol{2.2}$.] \textbf{$\boldsymbol{|\mu(f)|<q_f}.$} To prove this case, we can construct  matchings $\mu'$ that super weakly dominates $\mu$ in an analogous way to Case 2.2. in the proof of Theorem \ref{strong core = strongly stable}.
\end{enumerate}
By subcases 2.1 and 2.2,  $\mu \notin \mathcal{C}_{SS}.$ 

\end{itemize}

 Next, we prove that $\mathcal{SSS}\subseteq \mathcal{C}_{SS}.$ Assume that $\mu \in
\mathcal{SSS} $ and $\mu \notin\mathcal{C}_{SS}.$ Thus, there are a non-empty coalition $C$ and a matching $\mu'$ such that $\mu'$ super weakly dominates $\mu $ via  $C.$ First, we claim that $C\cap F \neq \emptyset.$ Otherwise,  $C\cap F = \emptyset$ implies, for each $w \in C,$ that $\mu'(w)=\emptyset$ (since $\mu'(C)=C$), and therefore $\emptyset=\mu'(w) R_w \mu(w),$ contradicting the individual rationality of $\mu.$ So take any $f\in C$. Since $\mu'$ super weakly dominated $\mu$ via $C$, $\mu'(f)R_f \mu(f).$ Now we claim that $\mu'(f)\setminus \mu(f)\neq \emptyset.$  Otherwise, $\mu'(f)\setminus \mu(f)=\emptyset$, $\mu'(f)R_f \mu(f)$, and $\mu'\neq \mu$ imply $\mu'(f)\subset\mu(f).$ Let $\mu(f)\setminus \mu'(f):= \{w_1,\ldots,w_\ell\}$. By individual rationality of $\mu$, $w_iP_f\emptyset$ for each  $i=1,\ldots,\ell.$ By responsiveness of $R_f$, 
$$
\mu(f)=\mu'(f) \cup\{w_1,\ldots,w_\ell\} P_f \mu'(f) \cup\{w_1,\ldots,w_{\ell-1}\}P_f \cdots P_f \mu'(f) \cup\{w_1\} P_f \mu'(f).
$$
Thus $\mu(f) P_f \mu'(f),$ contradicting $\mu'(f)R_f \mu(f)$. Therefore, $\mu'(f)\setminus \mu(f)\neq \emptyset.$ Define 
$$
\mu'(f)\setminus \mu(f):= \{w'_1,\ldots,w'_k\}
$$
where subscripts are chosen so that $w'_i R_f w'_{i+1}$ for each $i=1\ldots,k-1,$ and  
$$
\mu(f)\setminus \mu'(f):= \{w_1,\ldots,w_\ell\}
$$
where subscripts are chosen so that $w_i R_f w_{i+1}$ for each $i=1\ldots,\ell-1.$
Now, there are two cases to consider:
\begin{itemize}
\item[$\boldsymbol{1}.$]$\boldsymbol{|\mu'(f)|\leq|\mu(f)|}$. Then, $|\mu'(f)\setminus \mu(f)|\leq|\mu(f)\setminus \mu'(f)|.$ First, we claim that  
\begin{equation}\label{chorizo 0 theo 2}
w'_1 R_f w_\ell.
\end{equation} Otherwise, $$
w_1 R_f  \cdots R_fw_\ell P_f w'_1 R_f  \cdots R_f w'_k$$ 
and, by responsiveness of $R_f$ and the individual rationality of $\mu$, since $\ell \geq k$ 
\begin{equation}\label{chorizo 1 theo 2}
\mu(f) P_f \mu(f) \setminus \{w_\ell \} P_f \cdots P_f \mu(f) \setminus \{w_{k+1},\ldots,w_\ell \}.
\end{equation}
Define $\widetilde{W}:=\mu'(f)\cap \mu(f).$ Notice that $\mu(f) \setminus \{w_{k+1},\ldots,w_\ell \} = \widetilde{W}\cup \{w_1,\ldots,w_k\}.$ Again, by responsiveness of $R_f$, 
\begin{equation}\label{chorizo 2 theo 2}
\widetilde{W}\cup \{w_1,\ldots,w_k\} P_f \widetilde{W}\cup \{w_1,\ldots,w'_k\} P_f \cdots P_f \widetilde{W}\cup \{w'_1\ldots,w'_k\}.
\end{equation}
Notice that $\widetilde{W}\cup \{w'_1,\ldots,w'_k\}= \mu'(f).$ Thus, by \eqref{chorizo 1 theo 2} and \eqref{chorizo 2 theo 2}, $\mu(f)P_f \mu'(f)$. This contradicts that  $\mu'(f) R_f \mu(f)$ and, therefore, \eqref{chorizo 0 theo 2} holds and the claim is proven. Since $f\in C$ and $\mu'(C)=C$, $w'_1\in C$. As $\mu'$ super weakly  dominates $\mu$ via $C$, 
\begin{equation}\label{chorizo 3 theo 2}
f=\mu'(w'_1) R_{w'_1} \mu(w'_1).
\end{equation} 
If $|\mu(f)|=q_f$, by \eqref{chorizo 0 theo 2} and \eqref{chorizo 3 theo 2},   $(f,w'_1)$ is a super weakly blocking pair for $\mu$.  If $|\mu(f)|<q_f$, by \eqref{chorizo 0 theo 2}  we have $w'_1 R_f w_\ell$, and by the individual rationality of $\mu$, $w_\ell P_f \emptyset.$ Hence, $w'_1 P_f \emptyset$ and, therefore,   together with \eqref{chorizo 3 theo 2},   $(f,w'_1)$ is a super weakly blocking pair for $\mu$.

\item[$\boldsymbol{2}.$] $\boldsymbol{|\mu'(f)|>|\mu(f)|}$. Following a similar reasoning as in Case 2 of the proof of Theorem \ref{stable in core}, we can prove that  
\begin{equation}\label{chorizo 4 theo 2}
w'_1 P_f \emptyset.
\end{equation}

Notice that since $\mu'$ is a matching, $q_f \geq |\mu'(f)|>|\mu(f)|$. Since $f\in C$ and $\mu'(C)=C$, $w'_1\in C$. As $\mu'$ super weakly dominates $\mu$ via $C$, $f=\mu'(w'_1) R_{w'_1} \mu(w'_1)$ and, together with \eqref{chorizo 4 theo 2},  $(f,w'_1)$ is a super weakly blocking pair for $\mu$. 
\end{itemize}
 
Given that in each case we find a super weakly blocking pair for $\mu$, we contradict the fact  that $\mu \in\mathcal{SSS}.$ Therefore, $\mu \in\mathcal{C}_{SS}.$ 
\end{proof}

In order to compute solution concepts in matching models with indifferences, it is usual to study their relationship with solution concepts in matching models with strict preferences,  for which several algorithms are already available in the literature. Given a market $R,$ denote the set of all strict tie-breakings of $R$ by $\mathcal{L}(R).$   The following result provides a way to compute the super core of a market with indifferences in terms of  the stable matchings of all its associated strict markets.  

\begin{proposition}
For any  market $R,$ $\mathcal{C}_{SS}(R)=\bigcap_{P \in \mathcal{L}(R)}\mathcal{S}(P).$
\end{proposition}
\begin{proof}
Let $R$ be a many-to-one market. By Theorem \ref{super core = super stable}, $\mathcal{C}_{SS}(R)=\mathcal{SSS}(R).$ By Proposition 2 in \cite{irving2000hospitals}, $\mathcal{SSS}(R)=\bigcap_{P \in \mathcal{L}(R)}\mathcal{S}(P),$ 
and the result follows.  
\end{proof}

Our last result applies to the one-to-one model with indifferences, and provides a way to compute the core of a market with indifferences in terms of  the stable matchings of all its associated strict markets. 

\begin{proposition}\label{propo 1 a 1 core = union de estrictos}
For any one-to-one market $R$, $\mathcal{C}(R)=\bigcup_{P\in \mathcal{L}(R)}\mathcal{S}(P)$.
\end{proposition}
\begin{proof}
Let $R$ be a one-to-one market. By Proposition \ref{propo 1 a 1 core = estables}, $\mathcal{C}(R)=\mathcal{S}(R).$ By Proposition 1 in \cite{irving2000hospitals}, $\mathcal{S}(R)=\bigcup_{P \in \mathcal{L}(R)}\mathcal{S}(P),$ 
and the result follows. 
\end{proof}\begin{figure}[h!]

\centering

\begin{tikzpicture}[set/.style={fill=blue,fill opacity=0.12}]

\draw[set] (0,0) ellipse (2cm and 1cm);
 
\draw[set] (0,0) ellipse (3cm and 2cm);
 
\draw[set,
    ] (0,0) ellipse (4cm and 3cm);
 
\draw[set,
    ] (0,0) ellipse (5cm and 4cm);

\node[] at (0,0) {$\mathcal{C}_{SS}=\mathcal{SSS}$};    
\node[] at (0,-1.5) {$\mathcal{C}_{S}=\mathcal{SS}$};     
\node[] at (0,-2.5) {$\mathcal{S}$};
\node[] at (0,-3.5) {$\mathcal{C}$};
 
\end{tikzpicture}

\caption{\textsf{Inclusions among different cores and stability notions in the many-to-one setting.}}
\label{Inclusions}
\end{figure}
\section{Concluding remarks}\label{Concludings}

This paper presents, for a many-to-one model with indifferences, three notion of cores and its relationships with the three already known notions of stability. All the results and relations are summarized in Figure \ref{Inclusions}.

An interesting avenue for future research is to investigate whether our results are still valid in a many-to-one  model with indifferences in which firms' preferences are substitutable instead of responsive.

\end{document}